\documentclass[conference,10pt]{IEEEtran}
\ifCLASSINFOpdf
\else
\fi

\usepackage{eulervm}
\usepackage{graphicx}
\usepackage{array}
\usepackage{amsmath}
\usepackage{amssymb}
\usepackage{amsthm}
\usepackage{mathrsfs}
\usepackage{bbm}
\usepackage{bbold}
\usepackage{xfrac}
\usepackage{algorithmic}
\usepackage{algorithm}
\usepackage{upgreek}
\usepackage{multirow}
\usepackage{arydshln}
\usepackage{subfigure}
\usepackage{cite}
\usepackage{clrscode3e}
\usepackage{bm}
\usepackage{epstopdf}

\begin{document}
%
\title{Range Assignment for Power Optimization in Load-Coupled Heterogeneous Networks}


\author{\IEEEauthorblockN{Lei You$^1$, Lei Lei$^2$, and Di Yuan$^2$}
\IEEEauthorblockA{\small$^1$Information Engineering College,
Qingdao University, China\\}
\IEEEauthorblockA{\small$^2$Department of Science and Technology,
Link\"{o}ping University, Sweden\\}
\small\texttt{youleiqdu@gmail.com, \{lei.lei, di.yuan\}@liu.se}
}


%


\maketitle

\begin{abstract}
We consider the problem of transmission energy optimization via range assignment for Low Power Nodes (LPNs) in Long Term Evolution (LTE) Heterogenous Networks (HetNets). The optimization is subject to the load coupling model, where the cells interfere with one another. Each cell provides data service for its users so as to maintain a target Quality-of-Service (QoS). 
We prove that, irrespective the presence of maximum power limit or its value, operating at full load is optimal. 
We perform energy minimization by optimizing the association between User Equipments (UEs) and cells via selecting cell-specific offsets on LPNs. 
Moreover, the optimization problem is proved to be $\mathcal{NP}$-hard. 
We propose a tabu search algorithm for offset optimization (TSO). For each offset, TSO computes the optimal power solution such that all cells operate at full load. 
Numerical results demonstrate the significant performance improvement of TSO on optimizing the sum transmission energy, compared to the conventional solution where uniform offset is used for all LPNs.
\end{abstract}


%
\IEEEpeerreviewmaketitle

\theoremstyle{plain}
\newtheorem{definition}{Definition}
\newtheorem{theorem}{Theorem}
\newtheorem{lemma}{Lemma}
\newtheorem{proposition}{Proposition}
\newtheorem{postulation}{Postulation}

\section{Introduction}
Recently, Heterogeneous Networks (HetNets) are viewed as an attractive approach for expanding mobile network capacity \cite{Zhang:2013hv}, balancing the network load as well as alleviating the traffic burden.
A HetNet is a mix of both overlaying Macro Cells (MCs) and underlying small Low-Power Nodes (LPNs). MCs provide wide area data services. LPNs are deployed within traffic hotspots, offloading part of the traffic volume from MCs \cite{Hiltunen:2011}. LPNs are able to handle the growing geographical diversity of traffic and the increasing variation in local densities of user distribution. With the deployment of LPNs, HetNets can be thereby viewed as a way of meeting traffic demands and performance expectations, particularly in situations where traffic is concentrated in hotspots, or areas that cannot be suitably covered by MCs.

In a Long Term Evolution (LTE) HetNet is that, cells using the same frequency band interfere with one another such that a load coupling model should be considered\cite{Ericsson:2010LTE}.
We refer to the average level of usage of the time-frequency resource units as the load.
One challenge is that a base station consumes a significant fraction of the total end-to-end energy \cite{Ericsson:2010LTE}.
The cell load levels affect the sum transmission energy for serving User Equipments (UEs), and is thereby a key aspect of resource optimization in LTE HetNets. 
Within this context, cell selection is a crucial aspect for system performance.
Each UE follows the rule of selecting the cell with the best received signal power plus an offset, which is a common cell parameter.
A served UE may add the offset to its received power from a cell. Then, the signal strength of the cell can be `virtually' amplified to be the best for the UE. 
Thus the association between cells and UEs are influenced by the cell coverage control mechanism. 
\vskip -1pt
In the previous work \cite{Ho:2014icc,Ho:2014sub}, the power adjustment algorithm is given for all base stations to achieve full load, thus optimizing the sum transmission energy. 
With the same load coupling model, \cite{Ho:2014towc} provides a utility metric framework, where the utility function is maximized by optimizing the demand to be served. 
In \cite{Siomina:2012icc}, an optimization framework for load balancing in LTE HetNets is proposed, by means of cell range assignment using cell-specific offset.
In this paper, we focus on the problem of minimizing the sum power
used for transmission in LTE HetNets, while maintaining a target Quality-of-Service (QoS) for each UE.
The approach is to implement cell range assignment as in \cite{Siomina:2012icc}, thus changing the association between UEs and cells to achieve the minimal total energy consumption. Our contributions are as follows.\vskip -10pt
\begin{enumerate}
\item We propose a framework modeling LPN range optimization for HetNet energy minimization.
\item We prove that, independent of the maximum power limit, operating at full load leads to minimum energy consumption.
\item We consider the problem complexity, and prove its $\mathcal{NP}$ hardness.
\item We propose an algorithm for offset optimization based on tabu search, for which we numerically demonstrate its effectiveness and performance improvement.
\end{enumerate}
\vskip -2pt

The paper is organized as follows. In Section~\ref{sec:sys_mod} we introduce the HetNet and load coupling model. Under the model, we formally define the optimization problem of minimum sum energy in Section \ref{sec:optimization}. In Section~\ref{sec:full_load}, we prove the optimality of the full load independent of the maximum power limit. We then prove the problem is $\mathcal{NP}$-hard. Section~\ref{sec:algorithm} proposes the tabu-search algorithm for energy optimization via range assignment. Section~\ref{sec:numerical} shows the numerical results. Finally, the paper is concluded in Section~\ref{sec:conclusion}.
\vskip -6pt

\section{System Model}
\label{sec:sys_mod}

\subsection{HetNet System}
The sets of MCs and LPN cells are denoted by $\mathcal{I}_1$ and $\mathcal{I}_2$, respectively. Let $\mathcal{I}=\mathcal{I}_1\cup\mathcal{I}_2$, $n=|\mathcal{I}|$ and $m=|\mathcal{I}_2|$. The set of UEs is denoted by $\mathcal{J}$. 
Cell $i$ transmits with power $0< p_i\leq l_i$ per resource unit (in time and frequency, e.g., resource block in LTE), where $l_i$ is the power limit for each resource unit of cell $i$.  Notation $\nu_i$ represents the load of cell $i$ for data transmission, which is viewed as the portion of resource consumption of the cell. The network-wide load is given by the vector, $\bm{\nu}=[\nu_1,\nu_2,\ldots,\nu_n]^{\mathsf{T}}$.

The cells in $\mathcal{I}_1$ use zero offset (set $\mathcal{I}_1$ may also be empty), and cell-specific offset optimization only applies to the LPN cells in set $\mathcal{I}_2$. The candidate offset values compose set $\mathcal{S}$. We use $x_{is}\in\{0,1\}$ to denote whether or not cell $i\in\mathcal{I}_2$ uses offset level $s\in\mathcal{S}$, in dB. 
Each cell uses one of the candidates of the offset levels and thereby we have the constraint $\sum_{s\in\mathcal{S}}x_{is}=1$. 
We use $\bm{x}$ to denote the network-wide range assignment. That is, the association between UEs and cells, is represented as $\bm{u}(\bm{x})$. If cell $i\in\mathcal{I}$ is the serving cell of UE $j\in\mathcal{J}$, then $u_{ij}(\bm{x})=1$.
The resulting network-wide range assignment of $\bm{x}$, that is, the allocation of UEs to cells, is denoted by $\bm{u}(\bm{x})$, with $u_{ij}(\bm{x})=1$, if cell $i\in\mathcal{I}$ is the serving cell of UE $j\in\mathcal{J}$. 
The UE-cell association is determined by the control channel. 
The pilot power used to compute the association is fixed. 
For LPNs, the UE-LPN association can be adjusted by adding an offset $s\in\mathcal{S}$ to the pilot power.
Each UE can only be served by one cell at a time, so we have the constraint $\sum_{i\in\mathcal{I}}u_{ij}(\bm{x})=1$. For each $i\in \mathcal{I}$, the subset of UEs served by cell $i$ is related to the association $\bm{u}$, which is determined by the offset $\bm{x}$. So we denote the set of UEs served by cell $i$ as a function of offset $\bm{x}$, i.e., $\mathcal{J}_i(\bm{x})\triangleq\{j:u_{ij}(\bm{x})=1\}$.

\subsection{Load Coupling}

Next we consider the load coupling model for the HetNet. The load $\nu_{i}$, measures the fractional usage of resources in cell $i$. In LTE systems, the load can be viewed as the expected fraction of the time-frequency resources that are scheduled to deliver data. Suppose $u_{ij}=1$, then we model the SINR of user $j$ in cell $i$ as \cite{Siomina:2012icc,Siomina:2009ieee,Majewski:2010aict,Siomina:2012twc,Fehske:2012icc}. 
\vskip -3pt
\begin{equation}
\mathsf{SINR}_{ij}(\bm{\nu})=\frac{p_ig_{ij}}{\sum_{k\in\mathcal{I}\backslash\{i\}}p_kg_{kj}\nu_k+\sigma^2}
\label{eq:sinr}
\end{equation}
\vskip -3pt
In (\ref{eq:sinr}) $g_{kj},(k\neq i)$, represents the channel gain from the interfering cells. The noise power is denoted by $\sigma^2$. We denote the channel power gain from cell $i$ to user $j$ by $g_{ij}$. 
Load $\nu_k$ has the role of interference scaling, as it is intuitively interpreted as the probability that the served UEs of cell $i$ receives interference from cell $k$ on all resource units.
Thus, the expected interference with expectation taken over time and frequency for all transmissions, is denoted by $p_{k}g_{kj}\nu_{k}$.

We apply the cell load coupling function developed in the previous works \cite{Ho:2014towc,Ho:2014icc}. The achievable rate is given by $\tilde{r}_{ij} = B\log(1 + \mathsf{SINR}_{ij} )$ nat/s, where $B$ is the bandwidth of the resource unit and $\log$ is the natural logarithm. The demand of user $j$ on cell $i$ is $r_{ij}$ such that cell $i$ has to occupy $\nu_{ij}\triangleq r_{ij}/\tilde{r}_{ij}$ resource units. Assume that there are $M$ resource units available in total, then we get the load for the cell as $\nu_i=\sum_{j\in\mathcal{J}_i(\bm{x})}\nu_{ij}/M$ as follows.
\vskip -3pt
\begin{equation}
\nu_i=\frac{1}{MB}\sum_{j\in\mathcal{J}_i(\bm{x})}\frac{r_{ij}}{\log(1+\mathsf{SINR}_{ij}(\bm{\nu}))}\triangleq f_i(\bm{\nu}), 
~i\in\mathcal{I}
\label{eq:load}
\end{equation}
\vskip -3pt
We let $MB = 1$ without loss of generality. Then in (\ref{eq:load}), $r_{ij}$ is normalized by $MB$, the amount of effective time-frequency resources. In (\ref{eq:load}), $\mathcal{J}_i(\bm{x})$ is determined by the association between cells and UEs, i.e., $\bm{u}(\bm{x})$. We let $\bm{f^{\bm{u}(\bm{x})}}(\bm{\nu}) = [f_1(\bm{\nu}), \ldots , f_n(\bm{\nu})]^{\mathsf{T}}$. In vector form, we obtain the \emph{non-linear load coupling equation} ($\mathsf{NLCE}$) as in \cite{Ho:2014towc,Ho:2014icc}.
\vskip -3pt
\begin{equation}
\mathsf{NLCE:}~~\bm{\nu}=\bm{f^{\bm{u}(\bm{x})}}(\bm{\nu};\bm{r},\bm{p}),~\bm{0}<\bm{\nu}\leq\bm{1}
\label{eq:NLCE}
\end{equation}
\vskip -3pt
We remark that the function $\bm{f}$ is related to the cell-UE association $\bm{u}(\bm{x})$, which is induced by the range assignment $\bm{x}$. Also, there is dependence of the load $\bm{\nu}$ on the demand $\bm{r}$ and power $\bm{p}$. The equation cannot be readily solved in closed-form, for that the load $\bm{\nu}$ appears in both sides.
However, for a target load, it is shown in \cite{Ho:2014icc} that the power converges into a fixed-point solution in (\ref{eq:NLCE}) through an iterative process. 
Given load vector $\bm{\nu}$, let $\bm{p}(\bm{\nu})=[p_1(\bm{\nu}),p_2(\bm{\nu}),\ldots,p_n(\bm{\nu})]$ be the power solution in (\ref{eq:NLCE}), where for any $i\in\mathcal{I}$, $p_i(\bm{\nu})$ is the power of cell $i$ at the convergence. 

Note that there is no need to consider users with zero demand.
Therefore the QoS constraints are defined as $\bm{r}\geq \bm{d}_{min}$, where $\bm{d}_{min}$ is strictly positive.  Any load $\nu_i$ for cell $i$ must be positive if we combine the above assumption with $|\mathcal{J}_i(\bm{x})|\geq 1$. So we have $\bm{0}<\bm{\nu}\leq\bm{1}$.

\section{Energy Minimization Problem and Its Hardness}
\label{sec:optimization}

Under the load coupled HetNet model in Section \ref{sec:sys_mod}, the energy minimization problem is given by Problem $P0$. 
\vskip 10pt
$x_{is}$ = Range assignment for cell $i\in\mathcal{I}_2$ on offset level $s\in\mathcal{S}$
\begin{subequations}
\begin{alignat}{2}
P0:~~~ \min &\quad \bm{\nu}^{T}\bm{p} \\
 \textnormal{s.t.} &\quad 0<p_i\leq l_i,~~~~~~~~~\forall i\in\mathcal{I} \\
 &\quad \bm{r}\geq\bm{d}_{min} \\
 &\quad \bm{\nu}=\bm{f^{\bm{u}(\bm{x})}}(\bm{\nu};\bm{r},\bm{p}), \bm{0}<\bm{\nu}\leq\bm{1}\\
 &\quad \sum_{i\in \mathcal{I}}u_{ij}(\bm{x})=1,~~~~~\forall j\in\mathcal{J} \\
 &\quad u_{ij}\in\{0,1\},~~~~~~~~~\forall i\in\mathcal{I},j\in\mathcal{J}\\ 
 &\quad \sum_{s\in \mathcal{S}}x_{is}=1,~~~~~~~~~\forall i\in\mathcal{I}_2 \\
 &\quad x_{is}\in\{0,1\},~~~~~~~~~\forall i\in\mathcal{I}_2,s\in\mathcal{S}
\end{alignat}
\label{eq:p0}
\end{subequations}
\vskip -10pt
The objective is to minimize the sum transmission energy given by $\sum_{i=1}^{n}\nu_ip_i$.
We note that the product $\nu_ip_i$ measures the transmission energy used by cell $i$, because the load $\nu_i$ reflects the normalized amount of resource units used in time or frequency while $p_i$ is the power per resource unit. As mentioned earlier, the power and demand variables are strictly positive. The maximum power constraint is shown in (\ref{eq:p0}b). Constraint (\ref{eq:p0}c) is imposed so that the demand $\bm{r}$ satisfies the QoS constraint. Constraint (\ref{eq:p0}d) is the load coupling constraint. The association constraint between cells and UEs are shown in (\ref{eq:p0}e) and (\ref{eq:p0}f). Constraint (\ref{eq:p0}g) and (\ref{eq:p0}h) are imposed to satisfy that one cell can only be of one offset level.


\section{Theoretical Properties}
\label{sec:full_load}

The work in \cite{Ho:2014icc,Ho:2014sub} proves that full load is optimal for energy minimization, where there is no limit for power on each resource unit. Extending the result in \cite{Ho:2014icc,Ho:2014sub}, we show the optimality of the full load under the maximum power constraint. Further, we show that $P0$ is $\mathcal{NP}$-hard. 

\begin{lemma}

For any two load $\bm{\nu}$ and $\bm{\nu'}$ with $\bm{\nu'}\leq\bm{\nu}$, $\bm{p}(\bm{\nu})\leq \bm{p}(\bm{\nu'})$.
\label{lem:opt_full2} 
\end{lemma}
\begin{proof}
The case of $\bm{\nu}=\bm{\nu'}$ is trivial,  so we only consider the case of $\bm{\nu}>\bm{\nu'}$.
Suppose there exists at least one cell $i$ with $\nu'_i<\nu_i$.
Let $\bm{p}^0=\bm{p}(\bm{\nu})=[p^0_1, p^0_2, \dots, p^0_n]$.
To compute $\bm{p}(\bm{\nu'})$, we obtain the power sequence, e.g.,  $p_i^0, p_i^1,p_i^2, p_i^3, \dots$ for cell $i$, by the power adjustment algorithm proposed in \cite{Ho:2014icc}.
Suppose we now reduce cell $i$'s load from $\nu_i$ to $\nu'_i$.
Then we obtain a new load-power pair ($\nu'_i$, $p^1_i$) from  ($\nu_i$, $p^0_i$) for cell $i$.
According to Lemma 7 in \cite{Ho:2014icc}, $p^1_i>p^0_i$ and $p_i^1\nu'_i>p^0_i\nu_i$. 
This will cause the interference originated from cell $i$ to any other cell $k~(k\neq i)$ to be greater. As a result, the $\mathsf{SINR}$ for cell $k$ to all its served UEs will decrease. 
Then the corresponding load of cell $k$ in Equation (\ref{eq:load}) will thereby increase,  which is larger than the target load $\nu'_k$. 
To reach the target load  $\nu'_k$, the power adjustment algorithm  has to increase cell $k$'s power $p^0_k$ to a larger value $p_k^1$. 
At this stage, cell $k$'s interference to any other cell $h~(h\neq k, i)$ thereby increases. Thus the power $p^0_h$ in cell $h$ will increase as well. The same process repeats for all the cells in the algorithm iterations. The network-wide power and load will converge to a fix point, as proved in Lemma 1 in \cite{Ho:2014sub}. At convergence, $\bm{p}(\bm{\nu})<\bm{p}(\bm{\nu'})$ and the lemma follows.

\end{proof}


\begin{theorem}
Suppose for any cell $i$, $p_i(\bm{1})\leq l_i$, then $\bm{\nu}=\bm{1}$ and $\bm{r}=\bm{d}_{min}$ are optimal for energy minimization.
\label{thm:less_full}
\end{theorem}
\begin{proof}
The optimality of the minimum demand $\bm{d}_{min}$ is proved in \cite{Ho:2014icc}.
Consider a problem variant $\overline{P}0$ of Problem $P0$, where the power constraint (\ref{eq:p0}b) is relaxed to $p_i>0, ~\forall i\in\mathcal{I}$. For $\overline{P}0$, the full load optimality is proved by Theorem 1 in \cite{Ho:2014icc} for any given user association. 
For the user association under the optimal offset  in $P0$, if the power solution $\bm{p{(\bm{1})}}$ for $P0$ is less or equal than power limit $\bm{l}=[l_1, l_2, \dots, l_n]$, then $\bm{p{(\bm{1})}}$ is optimal for $P0$. 
Otherwise, if there exists at least one cell $i$ that $p_i(\bm{1})>l_i$, 
then by Lemma \ref{lem:opt_full2}, reducing any cell's load cannot decrease $p_i(\bm{1})$. 
In this case, there is no solution for $P0$. Hence, the full load $\bm{\nu}=\bm{1}$ is optimal for $P0$. In other words, constraint (\ref{eq:p0}d) in $P0$ can be replaced by Equation (\ref{eq:NLCE_full}).
\vskip -7pt
\begin{equation}
\bm{1}=\bm{f}(\bm{1};\bm{d}_{min}, \bm{p{(\bm{1})}})
\label{eq:NLCE_full}
\end{equation}
\end{proof}
\vskip -30pt
\begin{theorem}
$P0$ is $\mathcal{NP}$-hard.
\label{thm:np-hard}
\end{theorem}
\begin{proof}
The proof is under the assumption of full load for all MCs and LPNs, i.e., $\bm{\nu}=\bm{1}$. The basic idea is to reduce the Maximum Independent Set (MIS) problem to Problem $P0$. We construct a specific HetNet scenario. 
For each served UE, there is one potential MC and one potential LPN. 
Correspondingly, we consider a graph with $n$ nodes $(n\geq 2)$.
Then for each node $i$ in the graph, we set three corresponding elements MC $i$, LPN $i$ and UE $i$ in the specific HetNet scenario.
In total we have $n$ MCs, $n$ LPNs and $n$ UEs. 
For any $i$, we set the gain between MC $i$ and UE $i$ to be $\frac{1}{n^2}$, and the gain between LPN $i$ and UE $i$ to be $1$. In addition, we set the gain between LPN $i$ and other UEs that are neighbors to node $i$ in the graph to be a small positive real number $\epsilon$. Gain values other than the above two cases are negligible, treated as $0$. The power limits for LPNs and MCs are set to $1$ and $\infty$, respectively. The noise $\sigma^2$ and any user demand $d_{ij}$ are both set to $1$. The set of offset levels for LPNs is $\{-\infty,0\}$. If the offset for LPN $i$ is $-\infty$, the LPN $i$ serves no UE. Otherwise the offset of LPN $i$ is $0$, and thereby LPN $i$ is activated to serve UE $i$. From the reduction, we show the following two points: 1) The association between LPNs and UEs in the specific HetNet is corresponding to the feasible MIS solution. 2) It's always better to use LPN in the specific HetNet. 

The first point can be seen from the fact that it would never happen for any $i$, that both UE $i$ and its neighboring UE are simultaneously served by their respective LPNs, due to the interference generated from $i$'s neighbor LPNs. For the second point, we prove that for a sufficiently small $\epsilon$, activating more LPNs is always beneficial to reducing total power.
Suppose there are $k~(k<n)$ LPNs using offset $0$, then the total power of the LPNs is $k$ and there are $n-k$ MCs serving the remaining $n-k$ UEs. Consider the best-possible case, i.e., no interference but only noise $\sigma^2=1$. For power consumption of any of the $n-k$ MCs, the load coupling equation is $1.0=\sfrac{1}{\log_2(1+p\times\frac{1}{n^2}/1.0)}$, and we get $p=n^2$. With $n-k$ MCs, we have the total power for all MCs as $(n-k)n^2$. Therefore, if $k$ LPNs use offset $0$, for the best-possible power in total, denoted by $P_{k}$, we have 
\begin{equation}
P_{k}\geq k+(n-k)n^2
\label{eq:best_possible}
\end{equation}
Suppose there is a solution with $k+1$ LPNs using offset $0$, then the total power of the LPNs is $k+1$. For the $n-k-1$ MCs, consider the worst-possible case, i.e., for each of the corresponding UEs, there is interference from LPNs (at the neighboring nodes). Clearly, there cannot be more than $k+1$ interfering LPNs, each generating interference $\epsilon$. Then, for each of the $n-k-1$ MCs to serve the corresponding UE, we have the load couple equation (\ref{eq:load}) to be $1.0=\sfrac{1}{\log_2(1+\frac{p\times\sfrac{1}{n^2}}{(k+1)\epsilon+1.0})}$ and we have $p=n^2(k+1)\epsilon+n^2$. With all $n-k-1$ MCs, we have the total power as $(n-k-1)n^2(k+1)\epsilon+(n-k-1)n^2$. Therefore, if $k-1$ LPNs use offset $0$, for the worst-possible power in total, denoted by $\overline{P}_{k+1}$, we have 
\begin{equation}
\overline{P}_{k+1}\leq k+(n-k)n^2+[1-n^2+\epsilon(n-k-1)n^2(k+1)]
\label{eq:worst_possible}
\end{equation}
Combining the two equation (\ref{eq:best_possible}) and (\ref{eq:worst_possible}), for a sufficiently small $\epsilon$, we have 
\[\overline{P}_{k+1}<P_{k}
\]
which means that activating more LPNs is always beneficial to reduce total power. 

Consequently, solving the specified range optimization scenario of Problem $P0$ to optimum gives the maximum independent set of the graph.
\end{proof}

\section{algorithm for computing offset solution}
\label{sec:algorithm}

We propose an algorithm based on Tabu Search for Offset optimization (TSO), to minimize energy at full load.
The objective (\ref{eq:p0}a) of $P0$ is $\min~\sum_{i=1}^{n}p_i$ at full load. The objective function \cite{Lundgren:2010} is $g(\bm{x})\triangleq\sum_{i=1}^{n}p_i$ in TSO, where the vector $\bm{x}$ represents the offset levels, as same as that we mentioned earlier in Section \ref{sec:sys_mod}.
We let $x_i\in\{1,2,\ldots,|\mathcal{S}|\}$ so that $x_i=s$ means the LPN cell $i$ is assigned with the $s$th level offset, which is corresponding to $x_{is}=1$ in Problem $P0$. There are totally $m=|\mathcal{I}_2|$ elements in the offset solution vector $\bm{x}$ for LPNs. 
The neighborhood of $\bm{x}$ is defined as $N(\bm{x})=\left\{\bm{x'}\middle|\sum_{i=1}^{m}\left| x_i- x'_j\right|\leq 1\right\}$. Suppose the current best solution is $\bm{x^{*}}$, if $g(\bm{x^{*}})$ is not improved after $\alpha$ steps, then the termination rule will be triggered and the algorithm will stop. The tabu attribute is defined as a modification on the position of a solution vector. The tabu length is denoted by $\beta$. We set the aspiration rule as follows. Suppose $\bm{x^{*}}$ is the current best solution, $\forall$ $\bm{x}$, if $g(\bm{x})<g(\bm{x^{*}})$, 
then the solution $\bm{x}$ can be chosen even if the modification position in the variation $\bm{x^*}\rightarrow\bm{x}$ is given the tabu status in the table $T$.

\begin{algorithm}
\renewcommand{\algorithmicrequire}{\textbf{Given: }}
\renewcommand{\algorithmicensure}{\textbf{Output: }}
\begin{algorithmic} 
\caption{Tabu Search for Offset Optimization (TSO).}
\label{alg:tabu search}
\REQUIRE $\bm{d}_{min}$, $[l_1,l_2,\ldots,l_n]$, $\bm{\nu^*}=\bm{1}$\\
\ENSURE $\bm{x^{*}}$ (the optimal offset vector)\\
\begin{codebox}
\li $\bm{x^{*}}\gets \bm{x'}$
\li $\bm{p^{*}}\gets$ $\min\{\bm{p}:\bm{1}=\bm{f^{\bm{u}(\bm{x^{*}})}}(\bm{1};\bm{d}_{min},\bm{p})\}$; $k \gets 0$
\li $T$=
  \begin{tabular}{|c|c|c|c|}
  \hline
   $1$ & $2$ & $\cdots$ & $m$\\
   \hline
   $0$ & $0$ & $\cdots$ & $0$ \\
    \hline
  \end{tabular} 
\li \While $k\leq\beta$ \Do 
\li     $C\gets\phi$
\li     \For $\forall\bm{x}\in N(\bm{x'})$ \Do
\li         \If $T[i_{\bm{x'}\rightarrow\bm{x}}]=0$ \Then
\li             $C\gets C\cup\{\bm{x}\}$
            \End
        \End
\li     $\bm{x''}\gets\mathop{\arg\min}\limits_{\bm{x}\in C} g(\bm{x})$
\li     $T[i_{\bm{x'}\rightarrow\bm{x''}}]\gets\beta$
\li     $\bm{p''}\gets\min\{\bm{p}:\bm{1}=\bm{f^{\bm{u}(\bm{x''})}}(\bm{1};\bm{d}_{min},\bm{p})\}$
\li     \If $g(\bm{x''})<g(\bm{x^{*}})~\textbf{and}~\forall i~p''_i\leq l_i$  \Then
\li         $\bm{x^{*}}\gets\bm{x''}$; $\bm{p^{*}}\gets\bm{p''}$; $k\gets 0$
\li     \Else
\li         $k\gets k+1$
        \End
\li     \For $i\gets 1 \To m$ \Do
\li         \If $T[i]>0$ \Then
\li             $T[i]=T[i]-1$
            \End
        \End
\li     $\bm{x'}\gets\bm{x''}$
    \End
\li \Return $\bm{x^*}$            
\end{codebox}
\end{algorithmic} 
\end{algorithm}
TSO is given in Algorithm 1. We compute the optimal power vector $\bm{p}(\bm{1})$ for each offset, by using the iterative bi-section algorithm in \cite{Ho:2014icc}. 
For any two neighbored solution vectors $\bm{x'}$ and $\bm{x}$, we denote the modification position between the two vectors as $i_{\bm{x'}\rightarrow\bm{x}}$. In other words, suppose $\bm{x'}$ and $\bm{x}$ differ in the $j$th position, then $i_{\bm{x'}\rightarrow\bm{x}}=j$. 
In Line 1, a feasible solution $\bm{x'}$ is randomly chosen at the beginning and is assigned to be the current best solution $\bm{x^*}$. In Line 2, the power $\bm{p^{*}}$ at full load is computed. Also, $k$ is the count parameter that records the number of steps within which the current best solution $\bm{x^{*}}$ has not been improved. Line 3 initializes the tabu table $T$ by setting all the positions to be zero. Line 4--19 are the main loop that ends when the termination rule is triggered, i.e., $k>\beta$. The eligible neighbors set $C$ of the current solution vector $\bm{x}$ is calculated in Line 5--8. Then we choose the solution $\bm{x''}\in C$ that achieves the minimal value of $g$, as shown in Line 9. Regarding for the variation $\bm{x'}\rightarrow\bm{x''}$, the modification position $i_{\bm{x'}\rightarrow\bm{x''}}$ is given the tabu status in $T$, shown in Line 10. We thereby get the corresponding power vector $\bm{p''}$ for $\bm{x''}$ in Line 11. Line 12--15 state the update rule for the best solution. If the new solution $\bm{x''}$ is better and satisfies the power constraint (\ref{eq:p0}b), both $\bm{x^{*}}$ and $\bm{p^{*}}$ will be updated and the count variable $k$ will be set back zero, otherwise we increase $k$ by one. In Line 16--18, $T$ is refreshed, by decreasing all non-zero records by one. In Line 19, the current solution vector $\bm{x'}$ is updated by $\bm{x''}$, for the next iteration. When the termination rule is triggered,
the best solution $\bm{x^{*}}$ is returned, in line 20.

\section{numerical results}
\label{sec:numerical}

In this section, we numerically investigate the theoretical findings and evaluate performance in LTE HetNets. 
\subsection{Network Configuration}
\begin{figure}[!tbp]
  \centering
  \includegraphics[width=0.55\linewidth]{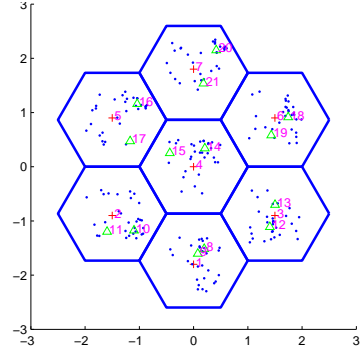}
    \caption{Network layout}
  \label{fig:network}
  \vskip -15pt
\end{figure}
The network layout is illustrated in \figurename~\!\ref{fig:network}, where the numbers indicate the cell IDs. There are 7 hexagonal serving cells, where the center is deployed with an MC (cross marker). For each MC, we randomly place two LPN cells (green triangles) in the corresponding hexagon. There are 21 cells (7 MCs and 14 LPN cells) in total. We generate 30 users (blue dots) for each MC (and also for its two LPN cells). All the users are distributed randomly and uniformly in each area. 
The HetNet operates at 2 GHz. Each resource unit follows the LTE standard of 180 KHz bandwidth and the bandwidth for each cell is 4.5 MHz.
The power limit per resource unit for MCs and LPNs are set to 200 mW and 50 mW, respectively. The noise power spectral density is set to -174 dBm/Hz. The channel gain $g_{ij}$ consists of path loss and shadowing fading, where the path loss follows the widely used COST-231-HATA model and the shadowing coefficients are generated by the log-normal distribution with 8 dB standard deviation. There are 11 offset levels for the LPN cells, as 0 dB, 1 dB, 2 dB,$\ldots$,10 dB. For TSO, we set $\alpha=10\times m=140$ and $\beta=\lceil\sqrt{m}\rceil=\lceil\sqrt{14}\rceil=4$. Note that the offset value only affects the UE-cell association and has no effect on the real transmit power or interference.


\subsection{Results of Offset Optimization}

The sum energy for different cases appears to grow exponentially fast as demand increases in \figurename~\ref{fig:plot1}. Beyond some demand value, the user demand may not be satisfiable. One can observe in \figurename~\ref{fig:plot2} that the power on each cell is far less than the power limit in the feasible solution. However, if we increase the user demand, then the power of a few cells may dramatically increase such that it will exceed the maximum constraint value. In this case, there is no solution to meet the demand target. Thus the power in \figurename~\ref{fig:plot1} is much lower than the maximum power constraint.

\begin{figure}[!tbp]
  \centering
  \includegraphics[width=0.82\linewidth]{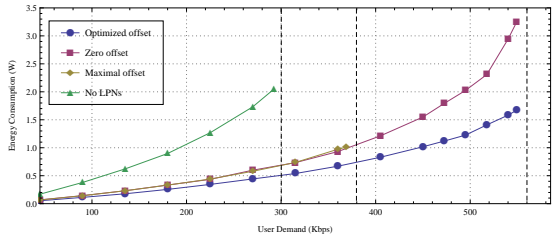}
    \caption{Energy consumption with respect to user demand}
  \label{fig:plot1}
\end{figure}

\begin{figure}[!tbp]
  \centering
  \includegraphics[width=0.82\linewidth]{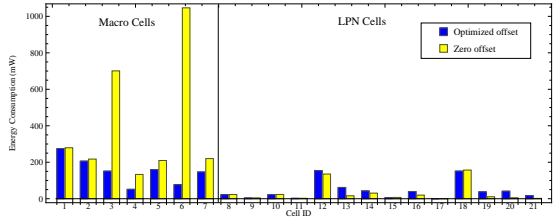}
    \caption{Energy consumption for optimized offset and zero offset in each cell.}
  \label{fig:plot2}
  \vskip -10pt
\end{figure}
In \figurename~\ref{fig:plot1}, we compare the Optimized Offset (OO) with Zero Offset (ZO) and Maximal Offset (MO, set to 10 dB). We also show the scenario of No LPN (NL) as the baseline. In the experiments, OO achieves the lowest sum transmission energy among all the three cases.
On average, the energy consumption for OO is about 31.9\% lower than that for ZO. When the demand is 550 Kbps, the energy consumption for OO is about 48.5\% lower than that for ZO. The worst performance with LPNs deployed is given by MO, since the transmit power for LPN cells are `virtually' amplified too large so that only a few UEs are assigned to the MCs. 

\figurename~\ref{fig:plot2} shows the transmission energy in each cell for ZO and OO. The user demand is set to 550 Kbps. For ZO, most users are served by MCs, so the majority of energy consumption is in MC 1--7. Specifically, for ZO, the energy consumption is very high in MC 3 and 6. In comparison, energy consumption in both MC 3 and 6 are significantly reduced in OO. This is because some UEs served by MC 3 and 6 are assigned to some LPNs in OO. 
One can observe that the energy consumption in OO stays at the almost same level with ZO in LPN cells.
By using the optimized offset to LPNs to adjust the cell-UE association, the UEs are assigned to more suitable cells. 
Thus the overall performance is significantly improved.

\section{conclusion}
\label{sec:conclusion}

We proposed a load coupling optimization framework for sum transmission energy in LTE HetNets via LPN range adjustment using cell-specific offsets. Under the maximum power constraint, we proved that full load is optimal for energy minimization. We further provided the insight that the optimization problem is $\mathcal{NP}$-hard. With the above theoretical properties, we proposed an algorithm TSO for optimizing the range assignment. For a scenario of HetNet deployment, LPN range optimization achieved by TSO leads to better sum transmission energy. 




%

\section{Acknowledgements}

This work has been supported by the EC Marie Curie
project MESH-WISE (FP7-PEOPLE-2012-IAPP: 324515) and the Link{\"o}ping-Lund Excellence Center in Information Technology (ELLIIT), Sweden.
The work of the second author has been supported by the China Scholarship Council (CSC).

\end{document}